\documentclass[reqno]{amsart}
\usepackage{amsmath,url}
\usepackage{amsthm}
\newcommand{\lbl}[1]{\label{#1}}
\newcommand{\arr}{\rightarrow}
\newcommand{\be}{\begin{eqnarray}}
\newcommand{\ee}{\end{eqnarray}}

\newcommand{\half}{\frac{1}{2}}
\newcommand{\ep}{\varepsilon}

\newcommand{\R}{{\mathbb R}}
\newcommand{\Rv}{{\mathbf R}}

\newcommand{\m}{{\mathbf m}}

\newcommand{\e}{{\mathbf e}}

\newcommand{\vv}{{\mathbf v}}
\newcommand{\p}{{\mathbf p}}

\newcommand{\q}{{\mathbf q}}
\newcommand{\Q}{{\mathbf Q}}

\newcommand{\K}{{\mathbf K}}

\newcommand{\M}{{\mathbf M}}

\newcommand{\X}{{\mathbf X}}
\newcommand{\bzero}{{\bf 0}}

\newcommand{\1}{{\bf 1}}

\newcommand{\weak}{\rightharpoonup}

\newtheorem{thm}{Theorem}[section]
\newtheorem{cor}[thm]{Corollary}
\newtheorem{lem}[thm]{Lemma}
\newtheorem{prop}[thm]{Proposition}

\theoremstyle{definition}
\newtheorem{defn}{Definition}[section]

\theoremstyle{remark}
\newtheorem{rem}{Remark}[section]

\numberwithin{equation}{section}

\begin{document}
\title
{Axisymmetry of critical points for the Onsager functional}
\author{J. M. Ball}
\address{Department of Mathematics, Heriot-Watt University, Edinburgh EH14 4AS, U.K.}
\email{jb101@hw.ac.uk}

\date{\today}
\begin{abstract}
 A simple proof is given of the classical result \cite{FatkullinSlastikov2005,LiuZhangZhang2005} that critical points for the Onsager functional with the Maier-Saupe molecular interaction are axisymmetric, including the case of stable critical points with an additional dipole-dipole interaction \cite{zhouetal2007}. The proof avoids  spherical polar coordinates, instead using an integral identity on the sphere $S^2$. For general interactions with absolutely continuous kernels the smoothness of all critical points is established, generalizing a result in \cite{vollmer2016} for the Onsager interaction. It is also shown that non-axisymmetric critical points exist for a wide variety of interactions including that of Onsager.
\end{abstract}
 \maketitle

\section{Introduction}
\setcounter{equation}{0}
In this paper we study critical points of  the Onsager free-energy functional
\be 
\label{1}
I(\rho)=\tau\int_{S^2}\rho(\p)\ln\rho(\p)\,d\p+\half\int_{S^2}\int_{S^2}k(\p\cdot\q)\rho(\p)\rho(\q)\,d\p\,d\q,\ee
where $\tau=k_B\theta$, $\theta$ is the absolute temperature and $k_B$ is Boltzmann's constant, so that $\tau>0$. Here $\rho$ represents the probability of finding a rod-like molecule with orientation $\p\in S^2$ in a homogeneous nematic liquid crystal.  The first term in \eqref{1} is entropic, while the second represents the potential  energy due to interactions between pairs of molecules but can also contain an entropic component. (For more details on the physical background see, for example, \cite{p38}.)

The kernel $k:[-1,1]\to\R$ is assumed to be continuous. That   $k$ depends on $\p,\q$ only through the scalar product $\p\cdot\q$ ensures that $I$ is $O(3)$ invariant, that is $I(\rho_\Rv)=I(\rho)$ for any $\Rv\in O(3)$, where $\rho_\Rv(\p):=\rho(\Rv\p)$. Important examples  are:
\begin{eqnarray*}
 &k(t)=-\kappa t^2, &\text{(Maier-Saupe)},\\ 
 &k(t)=-\sigma t,&\text{(dipolar)},\\
 &k(t)=-(\sigma t+\kappa t^2)&\text{(coupled dipolar/Maier-Saupe)},\\
&k(t)=\kappa \sqrt{1-t^2}&\text{(Onsager)},
\end{eqnarray*}
where $\kappa>0$, $\sigma>0$ are constants. Kernels   which are even in $t$ (i.e. $k(t)=k(-t)$), such as the Maier-Saupe and Onsager kernels,  correspond to liquid crystals exhibiting statistical head-to-tail symmetry of the molecules. Note that adding a constant to $k$ only changes $I(\rho)$ by a constant, so does not affect critical points. 

Critical points of $I$ are defined in Section \ref{crit} and shown to be probability density functions $\rho$ satisfying the Euler-Lagrange equation for \eqref{1}
\be 
\lbl{EL}
\rho(\p)=\frac{\psi(\p)}{\int_{S^2}\psi(\p)\,d\p },
\ee 
where 
\be 
\lbl{psidef}\nonumber
\psi(\p)=\exp (-\tau^{-1}(k*\rho) (\p))
\ee
and
\be 
\lbl{conv}\nonumber
(k*\rho) (\p):=\int_{S^2}k(\p\cdot\q)\rho(\q)\,d\q.
\ee
The $O(3)$ invariance of $I$ implies that the set of critical points is $O(3)$ invariant, that is if $\rho$ is a critical point, so is $\ \rho_\Rv$ for any $\Rv\in O(3)$. The isotropic state $\rho_{\rm iso}(\p):=\frac{1}{4\pi}$ is a critical point for any $k$. If the kernel is absolutely continuous on $[-1,1]$ then we show (Theorem \ref{smoothness}) that all critical points are smooth, generalizing a result of Vollmer \cite{vollmer2016} for the Onsager kernel.

A probability density function $\rho$ is {\it axisymmetric} if $\rho(\p)=f(\p\cdot\e)$ for some $\e\in S^2$ and function $f:[-1,1]\to\R$. For the dipolar potential it is obvious that any critical point $\rho$ is axisymmetric because  \eqref{EL} already expresses $\rho$ as a function of $\m :=\int_{S^2}\p\rho(\p)\,d\p$
\be 
\lbl{axidipolar}\nonumber
\rho(\p)=\frac{\exp (\alpha\p\cdot\m)}{\int_{S^2}\exp (\alpha\q\cdot\m)\,d\q},
\ee
where $\alpha:=\frac{\sigma}{\tau}$.
However for the Maier-Saupe potential the axisymmetry of all critical points is not at all obvious, and was first proved independently by different methods in the fundamental papers of Fatkullin \& Slastikov \cite{FatkullinSlastikov2005} and Liu, Zhang \& Zhang \cite{LiuZhangZhang2005} (see also \cite{liu2019note}), who gave a full description of solutions in terms of the parameter $\beta:=\frac{\kappa}{\tau}$, thus providing a complete picture of the isotropic-nematic phase transition for this model. A further proof of these results was given in \cite{zhouetal2005}. 

In this paper a different and much simpler proof of the axisymmetry of critical points for the Maier-Saupe potential is given (see Theorem \ref{MSaxithm}). The idea is not to use spherical polars, because these desymmetrize the analysis by the choice of the polar axis, but rather to work in Cartesian coordinates and use an integral identity on the sphere. The method also works for stable critical points for the coupled dipolar/Maier-Saupe potential (see Theorem \ref{caxi}), recovering a result of Zhou et al \cite{zhouetal2007}. It is further shown (Theorem \ref{nonaxithm}) that non-axisymmetric critical points exist for a wide variety of kernels, including that for the Onsager interaction, the idea being to minimize $I$ among $\rho$ having cubic symmetry.

\section{Critical points}
 \lbl{crit}
We denote by
$${\mathcal P}=\{\rho\in L^1(S^2):\rho\geq 0, \int_{S^2}\rho(\p)\,d\p=1\}$$
the set of   probability density functions on $S^2$. Because of the unilateral constraint $\rho\geq 0$ and the singular behaviour of $\rho\ln\rho$ at $\rho=0$, a little care is needed in defining critical points. The definition below is meaningful without supplementary regularity hypotheses on $\rho$.
\begin{defn}
We say that $\rho\in\mathcal P$ is a {\it critical point} of $I$ if $I(\rho)<\infty$ and for some $\sigma>0$
\be 
\lbl{critpt}\limsup_{\lambda\to 0+}\frac{I(\rho+\lambda(\rho'-\rho))-I(\rho)}{\lambda}\geq 0\nonumber
\ee
for all $\rho'\in \mathcal P$ with $I(\rho')<\infty$ and $\|\rho'-\rho\|_{L^\infty(S^2)}\leq\sigma$.
\end{defn}
\begin{defn}
\lbl{localmin}
 $\rho\in\mathcal P$  is {\it an $L^\infty$ local minimizer of} $I$ if for some $\sigma>0$
$$I(\rho)\leq I(\tilde\rho)\text{  for all }\tilde\rho\in \mathcal P\text{ with }\|\rho-\tilde\rho\|_{L^\infty(S^2)}\leq\sigma.$$
\end{defn}
\noindent Thus any $L^\infty$ local minimizer of $I$ is a critical point.
\begin{rem}
The existence of at least one absolute minimizer of $I$ in $\mathcal P$ follows easily from the direct method of the calculus of variations (see Step 1 of the proof of Theorem \ref{nonaxithm}).
\end{rem}
\begin{prop}
\lbl{critprop}
$\rho\in\mathcal P$ is a critical point of $I$ if and only if $\rho$ is continuous, bounded away from zero, and satisfies
\be 
\lbl{EL1}
\tau(1+\ln\rho(\p))+(k*\rho)(\p)=C, \;\;\p\in S^2
\ee
for some constant $C$, i.e. \eqref{EL} holds.
\end{prop}
\begin{proof}
We use the  method of \cite{u9}, which is in turn based on \cite{p10}. We write
\be 
\lbl{Idecomp}
I(\rho)=E(\rho)+K(\rho),\nonumber
\ee
where
$$E(\rho):=\int_{S^2}\tau\rho(\p)\ln\rho(\p)\,d\p,\;\;K(\rho):=\half\int_{S^2}\int_{S^2}k(\p\cdot\q)\rho(\p)\rho(\q)\,d\p\,d\q.$$
Let $\rho,\rho'\in\mathcal P$ with $I(\rho)<\infty$, $I(\rho')<\infty$ and $\|\rho-\rho'\|_{L^\infty(S^2)}\leq\sigma$. Note that $K(\rho), K(\rho')$ are well defined and finite, so  that  $E(\rho), E(\rho')$ are finite.  Also
\be 
\lbl{Kdiff}
\lim_{\lambda\to 0+}\frac{K(\rho+\lambda(\rho'-\rho))-K(\rho)}{\lambda}=\int_{S^2}(k*\rho)(\rho'-\rho)\,d\p.
\ee
 Furthermore
\be 
\frac{E(\rho+\lambda(\rho'-\rho))-E(\rho )}{\lambda}
=\int_{S^2}\frac{\eta( \rho+\lambda(\rho'-\rho))-\eta(\rho)}{\lambda}\,d\p,\label{jb1}
\ee
where $\eta(\rho):=\tau\rho\ln\rho$. Since $\eta$ is convex, the integrand in \eqref{jb1} is nondecreasing in $\lambda$.
If $\rho\in\mathcal P$ is a critical point then, given $\delta>0$, by assumption there exists a decreasing sequence $\lambda_j\to 0+$ such that 
$$\frac{E(\rho+\lambda_j(\rho'-\rho))-E(\rho )}{\lambda_j}+\frac{K(\rho+\lambda_j(\rho'-\rho))-K(\rho)}{\lambda_j}\geq -\delta.$$
In particular, by \eqref{Kdiff}, \eqref{jb1} the integrals
$$\int_{S^2}\frac{\eta( \rho+\lambda_j(\rho'-\rho))-\eta(\rho)}{\lambda_j}\,d\p$$
are bounded below, so that 
  by  monotone convergence and the arbitrariness of $\delta$ we conclude that 
\be 
\label{jb2}
\int_{S^2}[\tau(1+\ln \rho )+(k*\rho)](\rho'- \rho)\,d\p\geq 0,
\ee
where $\ln 0:=-\infty$. 

 Define $S:=\{\p\in S^2:\rho(\p)\leq\frac{1}{8\pi}\}$, and denote by $|S|$ the two-dimensional Hausdorff measure of $S$, so that $|S|<4\pi=|S^2|$. For $s>0$ sufficiently small, let
$$\rho'(\p)=\left\{\begin{array}{ll}\rho(\p)+s,&\p\in S,\\ \rho(\p)-s\frac{|S|}{4\pi-|S|},&\p\not\in S.\end{array}\right.$$
Then $\rho'\in \mathcal P$ with $\|\rho'-\rho\|_{L^\infty(S^2)}\leq\sigma$, and  we deduce from \eqref{jb2} that $s\int_S\ln\rho\,d\p$ is bounded below, and hence $\rho (\p)>0$ for a.e. $\p\in S^2$.

 For $\ep>0$ define
$$S_\ep=\{\p\in S^2:\ep<\rho(\p)<\ep^{-1}\}.$$
Let $u\in L^\infty(S_\ep)$ with $\int_{S_\ep}u(\p)\,d\p=0$ and define
\be 
\lbl{rhoprime}\nonumber
\rho'(\p):=\left\{\begin{array}{ll}\rho(\p)+s\,u(\p),& \p\in S_\ep,\\
\rho(\p),&\p\not\in S_\ep,
\end{array}\right.
\ee
for $|s|$ sufficiently small. Then $\rho'\in\mathcal P$ with $\|\rho'-\rho\|_{L^\infty(S^2)}\leq\sigma$, and so from \eqref{jb2} we deduce that
\be 
\lbl{epEL}\nonumber
\int_{S_\ep}(\tau(1+\ln \rho )+(k*\rho))u\,d\p= 0,
\ee
and hence that for some constant $C_\ep$
\be 
\lbl{epEL2}\nonumber
\tau(1+\ln(\rho))+(k*\rho)(\p)=C_\ep \text{  for a.e. }\p\in S_\ep.
\ee
Since $S_\ep$ increases as $\ep$ decreases, and since $\rho(\p)>0$ a.e. implies that $\cup_{\ep>0}S_\ep=S^2$, the constant $C_\ep$ is independent of $\ep$ and we have that 
\be 
\lbl{EL2}\nonumber
\tau(1+\ln\rho(\p))+(k*\rho)(\p)=C\text{  for a.e. }\p\in S^2
\ee
for a constant $C$. Since $(k*\rho)(\p)$ is continuous in $\p$ it follows that $\rho$ has a continuous representative which is bounded away from zero, and that \eqref{EL1} holds as required.
 
Conversely, if $\rho$ is continuous and bounded away from zero, if \eqref{EL1} holds, and if $\|\rho'-\rho\|_{L^\infty(S^2)}\leq\sigma$, then it follows easily that 
$$\lim_{\lambda\to 0+} \frac{I(\rho+\lambda(\rho'-\rho))-I(\rho)}{\lambda}=0,$$
so that $\rho$ is a critical point. 
\end{proof}
\begin{thm}
\lbl{smoothness}
If $k\in W^{1,1}(-1,1)$ $($i.e. $k$ is absolutely continuous on $[-1,1]$$)$ then any critical point $\rho$ is smooth.
\end{thm}
\begin{proof}
We simplify and extend the method of Vollmer \cite[Proposition 45]{vollmer2016} for the Onsager interaction kernel. For $t\in\R$ and $i=1,2,3$, define the skew matrices
$$\K_1=\left(\begin{array}{ccc}0&0&0\\0&0&-1\\
0&1&0\end{array}\right),\; \K_2=\left(\begin{array}{ccc}0&0&-1\\0&0&0\\
1&0&0\end{array}\right),\; \K_3=\left(\begin{array}{ccc}0&-1&0\\1&0&0\\0&0&0\end{array}\right),$$
and $\Rv_i(t):=\exp (\K_it)$, so that, for example,
$$\Rv_1(t)=\left(\begin{array}{rrr}1&0&0\\
0&\cos t&-\sin t\\0&\sin t&\cos t\end{array}\right)$$
represents a rotation through an angle $t$ around the $x_1$-axis. Since the vectors $\frac{d}{dt}\Rv_i(t)\p|_{t=0}=\K_i\p$ span the tangent space to $S^2$ at $\p$, it follows that if $f\in C^0(S^2)$ and $r\geq 1$ then (see \cite{garrett})  $f\in C^r(S^2)$ if and only if 
$$\left(\prod_{i=1}^s\X_{j_i}\right)f\in C^0(S^2)\text{ for all } j_i\in\{1,2,3\}, 1\leq s\leq r,$$
 where $\X_if(\p):=\left.\frac{d}{dt}f(\Rv_i(t)\p)\right|_{t=0}$.
\begin{lem}
\lbl{convlem}
If $h\in C^0(S^2)$ then $\X_i(k*h)\in C^0(S^2)$. If $h\in C^1(S^2)$ then
\be 
\lbl{conveq}\nonumber
\X_i(k*h)=k*(\X_ih).
\ee
\end{lem}
\begin{proof}[Proof of Lemma.]
Let $h\in C^0(S^2)$. Let $k_j\in C^\infty([-1,1])$ with $k_j\to k$ in $W^{1,1}(-1,1)$. Then 
\be 
\lbl{15}
\int_{S^2}k_j(\Rv_i(t)\p\cdot\q)h(\q)\,d\q-\int_{S^2}k_j(\p\cdot\q)h(\q)\,d\q&&\\
&&\hspace{-1.5in}=\int_0^t\int_{S^2}k_j'(\Rv_i(s)\p\cdot\q)\K_i\Rv_i(s)\p\cdot\q\, h(\q)\,d\q\,ds.\nonumber
\ee
We pass to the limit $j\to\infty$ in \eqref{15}. The integrals on the left-hand side converge to the corresponding values for $k$ by bounded convergence. For the right-hand side we have that for some constant $C$
\begin{eqnarray*}
\int_0^t\int_{S^2}\left|\left(k_j'(\Rv_i(s)\p\cdot\q)-k'(\Rv_i(s)\p\cdot \q)\right)\K_i\Rv_i(s)\p\cdot\q\,h(\q)\right|\,d\q\,ds&&\\
&&\hspace{-2.5in}\leq C\int_0^t\int_{S^2}\left|k_j'(\Rv_i(s)\p\cdot\q)-k'(\Rv_i(s)\p\cdot\q)\right|d\q\,ds\\
&&\hspace{-2.5in}=C\int_0^t\int_{S^2}\left|k_j'(\p\cdot\q)-k'(\p\cdot\q)\right|d\q\,ds\\
&&\hspace{-2.5in}=2\pi C\int_0^t\int_{-1}^1\left|k_j'(x)-k'(x)\right|\,dx\,ds \arr 0,
\end{eqnarray*}
where the third line is obtained by changing variables $\q\to\Rv_i(s)\q$ on the sphere. Hence we obtain
\be 
\lbl{16}\nonumber
\int_{S^2}k(\Rv_i(t)\p\cdot\q)h(\q)\,d\q-\int_{S^2}k(\p\cdot\q)h(\q)\,d\q&&\\
&&\hspace{-1.5in}=\int_0^t\int_{S^2}k'(\Rv_i(s)\p\cdot\q)\K_i\Rv_i(s)\p\cdot\q\, h(\q)\,d\q\,ds.\nonumber
\ee
Similarly, if  $t_j\to 0$ then
\begin{eqnarray*}
\int_{S^2}k'(\Rv_i(t_j)\p\cdot\q)\K_i\Rv_i(t_j)\p\cdot\q\,h(\q)\,d\q&&\\
&&\hspace{-1.5in}=\int_{S^2}k'(\p\cdot\q)\K_i\Rv_i(t_j)\p\cdot\Rv_i(t_j)\q\,h(\Rv_i(t_j)\q)\,d\q,\\
&&\hspace{-1.5in}\arr \int_{S^2}k'(\p\cdot\q)\K_i\p\cdot\q\,h(\q)\,d\q,
\end{eqnarray*}
so that $(k*h)(\Rv_i(t)\p)$ is differentiable at $t=0$ with derivative
\be 
\lbl{17}
\left.\frac{d}{dt}(k*h)(\Rv_i(t)\p)\right|_{t=0}=\int_{S^2}k'(\p\cdot\q)\K_i\p\cdot\q\,h(\q)\,d\q.
\ee
The right-hand side of \eqref{17} belongs to $C^0(S^2)$ because if $\p_j\to\p$ we can write $\p_j=\Q_j\e$ for some $\e\in S^2$ and $\Q_j\to\Q$ in $SO(3)$, so that 
\begin{eqnarray*}
\int_{S^2}k'(\p_j\cdot\q)\K_i\p_j\cdot\q\,h(\q)\,d\q&=&\int_{S^2}k'(\e\cdot \q)\K_i\Q_j\e\cdot\Q_j\q\,h(\Q_j\q)\,d\q\\
&\to&\int_{S^2}k'(\e\cdot\q)\K_i\Q\e\cdot\Q\q\,h(\Q\q)\,d\q\\
&=&\int_{S^2}k'(\p\cdot\q)\K_i\p\cdot\q\,h(\q)\,d\q.
\end{eqnarray*}
Hence $\X_i(k*h)\in C^0(S^2)$.

If in addition $h\in C^1(S^2)$ then
\begin{eqnarray*}
\X_i(k*h)(\p)&=&\frac{d}{dt}\left.\int_{S^2}k(\Rv_i(t)\p\cdot\q)h(\q)\,d\q\right|_{t=0}\\
&=&\frac{d}{dt}\left.\int_{S^2}k(\p\cdot\q)h(\Rv_i(t)\q)\,d\q\right|_{t=0}\\
&=&\int_{S^2}k(\p\cdot\q)\left.\frac{d}{dt}h(\Rv_i(t)\q)\right|_{t=0}d\q,
\end{eqnarray*}
as required.
\end{proof}
Continuing the proof of Theorem \ref{smoothness}, we show by induction that $\rho\in C^r(S^2)$ for any $r$. This is true for $r=0$, so suppose that $r\geq 1$ and $\rho\in C^{r-1}(S^2)$. Then since $\rho$ is bounded away from zero $\ln\rho\in C^{r-1}(S^2)$. Let $j_i\in\{1,2,3\}$, $i=1,\ldots, r$. Then  from \eqref{EL1} 
$\left(\prod_{i=2}^r\X_{j_i}\right)(k*\rho)\in C^0(S^2)$, so that, by repeated use of Lemma \ref{convlem}, 
$$\left(\prod_{i=2}^r\X_{j_i}\right)(k*\rho)=k*\left(\prod_{i=2}^r\X_{j_i}\right)\rho. $$
Hence, again by Lemma \ref{convlem}, $$\X_{j_1}\left(k*\left(\prod_{i=2}^r\X_{j_i}\right)\rho\right)=\left(\prod_{i=1}^r\X_{j_i}\right)(k*\rho)\in C^0(S^2).$$ But then by \eqref{EL1} $\left(\prod_{i=1}^r\X_{j_i}\right)\ln\rho\in C^0(S^2)$. Thus   $\ln\rho\in C^r(S^2)$ and hence $\rho\in C^r(S^2)$. 
\end{proof}

\section{Axisymmetry}
\lbl{axi}
\subsection{Axisymmetry for the Maier-Saupe interaction}
\lbl{MSax}

For the Maier-Saupe interaction it is convenient to use  the orthonormal eigenbasis $\{\e_i\}$ of 
the second moment tensor
\be 
\lbl{M}\nonumber
\M=\int_{S^2}\p\otimes\p\,\rho(\p)\,d\p,
\ee 
so that 
\be 
\lbl{M1}\nonumber
\M=\sum_{i=1}^3\gamma_i\e_i\otimes\e_i,
\ee 
where the $\gamma_i>0$ are the eigenvalues of $\M$. Writing $\p=\sum_{i=1}^3p_i\e_i$  the Euler-Lagrange equation \eqref{EL} then takes the form
\be 
\lbl{MSEL}
\rho(\p)=Z^{-1}\exp\left(\beta\sum_{j=1}^3\gamma_jp_j^2\right),
\ee
where  
\be 
\lbl{Z}\nonumber
Z:=\int_{S^2}\exp\left(\beta\sum_{i=1}^3\gamma_jp_j^2\right)\,d\p.
\ee
Thus we have to solve the equations 
\be 
\lbl{gamma}
\int_{S^2}p_i^2 \exp\left(\beta\sum_{i=1}^3\gamma_jp_j^2\right)=Z\gamma_i,\;\;i=1,2,3,
\ee
for the $\gamma_i$, and to prove axisymmetry is equivalent to showing that there are no solutions with the $\gamma_i$ all different.
\begin{lem}
\lbl{div}
If $u=u(\p)$ is a smooth function of $\p\in \R^3$ then 
\be 
\lbl{div1}\nonumber
\int_{S^2}\left(p_1\frac{\partial u}{\partial p_3}-p_3\frac{\partial u}{\partial p_1}\right)\,d\p=0.
\ee
\begin{proof}
Apply the divergence theorem on the unit ball $B(0,1)$ to the divergence-free vector field $\vv (\p)=\left(\frac{\partial u}{\partial p_3},0,-\frac{\partial u}{\partial p_1}\right)$ .
\end{proof}
\end{lem}
\begin{lem}
\lbl{expineq}
$$ x(e^x-1)>0\text{  for }x\neq 0.$$
\end{lem}
\begin{proof} This is elementary. \end{proof}
\begin{thm}
\lbl{MSaxithm}
All critical points for the Maier-Saupe interaction are axisymmetric.
\end{thm}
\begin{proof}
Suppose for contradiction that $\rho$ is a critical point with the $\gamma_i$ all different. We apply Lemma \ref{div} with $u(\p)=p_1p_3\,\rho(\p)$. Thus by \eqref{MSEL}
\be 
\lbl{10}\nonumber
\int_{S^2}(p_1^2-p_3^2)\rho(\p)\,d\p=2\beta(\gamma_1-\gamma_3)\int_{S^2}p_1^2p_3^2\,\rho(\p)\,d\p,
\ee
and hence, since $\gamma_1\neq\gamma_3$,
\be
\lbl{11}
2\beta\int_{S^2}p_1^2p_3^2\,\rho(\p)\,d\p= 1.
\ee
Swapping 2 and 3 (thus using that $\gamma_1\neq\gamma_2$) and subtracting the resulting equation from \eqref{11}, we obtain
\be
\lbl{12}
\int_{S^2}p_1^2(p_3^2-p_2^2)\,\exp\left(\beta\sum_{j=1}^3\gamma_jp_j^2\right)\,d\p= 0.
\ee
Interchanging $p_2$ and $p_3$ as integration variables on the sphere, we deduce that
\be 
\lbl{13}
\int_{S^2}p_1^2(p_2^2-p_3^2)\,\exp\left(\beta\left(\gamma_1p_1^2+\gamma_3 p_2^2+\gamma_2p_3^2\right)\right)\,d\p= 0,
\ee
so that adding \eqref{12}, \eqref{13}  we have that
\be 
\lbl{14}\nonumber
\int_{S^2}p_1^2(p_3^2-p_2^2)\,\exp\left(\beta\sum_{j=1}^3\gamma_jp_j^2\right)\left(1-\exp\left[\beta(\gamma_3-\gamma_2)(p_2^2-p_3^2)\right]\right)\,d\p=0.
\ee
Multiplying by $\gamma_3-\gamma_2$ and using Lemma \ref{expineq} we see that since $\gamma_2\neq\gamma_3$ the integrand is strictly positive for a.e. $\p\in S^2$, a contradiction.

\end{proof}
\subsection{Axisymmetry for the coupled dipolar/Maier-Saupe interaction}
\lbl{dipMS}

Zhou, Wang, Wang \& Forest   \cite{zhouetal2007} show that stable critical points are axisymmetric for the coupled dipolar/Maier-Saupe interaction. We show that the method in Section \ref{MSax} also works in this case.  We denote by 
\be 
\lbl{polarity}\nonumber
\m:=\int_{S^2}\p\,\rho(\p)\,d\p 
\ee 
the first moment of $\rho$,  the {\it polarity vector}. For the coupled dipolar/Maier-Saupe interaction critical points $\rho$ are solutions to
\be 
\lbl{EL3}
\rho(\p)=Z^{-1}\exp\left(\alpha\m\cdot\p+\beta\sum_{j=1}^3\gamma_jp_j^2\right),
\ee 
where $Z=\displaystyle\int_{S^2}\exp\left(\alpha\m\cdot\p+\beta\sum_{j=1}^3\gamma_jp_j^2\right)\,d\p$. 
 
Note that any critical point for the Maier-Saupe potential satisfies $\m=\bzero$ and thus is a solution to \eqref{EL3}, and hence  by Theorem \ref{MSaxithm} is axisymmetric. Thus for the purpose of proving axisymmetry we may assume that $\m\neq \bzero$. Following \cite{zhouetal2007} we then choose the eigenbasis of $\M$ so that $m_3=\m\cdot\e_3>0$. With this choice it is  proved in \cite{jietal2006}, \cite[Theorem 1]{zhouetal2007} that $m_1=m_2=0$, i.e. $\m$ is coaxial with an eigenvector of $\M$. It is also shown in \cite[Theorem 2]{zhouetal2007} that for $\rho$ to have nonzero polarity vector $\m$ we must have that $\alpha\gamma_3>1$; in particular, since each $\gamma_i<1$, all critical points have $\m=\bzero$ when $\alpha\leq 1$ (i.e. when the dipole-dipole interaction is sufficiently weak), a result proved in \cite{jietal2006}.

It is proved in \cite[Theorem 4]{zhouetal2007}, that any $L^\infty$ local minimizer $\rho$ of $I$ satisfies
\be 
\lbl{18}
\gamma_3>\max(\gamma_1,\gamma_2).
\ee
We will show that any critical point is axisymmetric under the weaker requirement that 
\be 
\lbl{18a}
(\gamma_3-\gamma_1)(\gamma_3-\gamma_2)>0.
\ee

\begin{thm}
\lbl{caxi}
If \eqref{18a} holds then $\gamma_1=\gamma_2$ in the selected coordinate system, so that $\rho=\rho(p_3)$ is axisymmetric. In particular any $L^\infty$ local minimizer of $I$  is axisymmetric.
\end{thm}
\begin{proof}
Assume for contradiction that  $\gamma_1\neq\gamma_2$. We again apply Lemma \ref{div} with $u(\p)=p_1p_3\,\rho(\p)$, to obtain (since $m_1=0$)
\be 
\lbl{19}\nonumber
(\gamma_3-\gamma_1)\left(2\beta\int_{S^2}p_1^2p_3^2\rho(\p)\,d\p-1\right)=-\alpha m_3\int_{S^2}p_1^2p_3\rho(\p)\,d\p.
\ee
Similarly
\be 
\lbl{20}\nonumber
(\gamma_3-\gamma_2)\left(2\beta\int_{S^2}p_2^2p_3^2\rho(\p)\,d\p-1\right)=-\alpha m_3\int_{S^2}p_2^2p_3\rho(\p)\,d\p.
\ee
Hence
\be 
\lbl{21}
2\beta\int_{S^2} p_3^2(p_1^2-p_2^2)\rho(\p)\,d\p=\alpha m_3\left(\frac{a}{\gamma_3-\gamma_2}-\frac{b}{\gamma_3-\gamma_1}\right),
\ee
where
\be 
\lbl{22}\nonumber
a:=\int_{S^2}p_2^2p_3\rho(\p)\,d\p,\;\;\;b:=\int_{S^2}p_1^2p_3\rho(\p)\,d\p.
\ee
Note that
\be 
\lbl{23}\nonumber
(\gamma_2-\gamma_1)\int_{S^2}p_3^2(p_1^2-p_2^2)\rho(\p)\,d\p=\half\int_{S^2}p_3^2f(\p)g(\p)\,d\p, 
\ee
where $f(\p):=(\gamma_2-\gamma_1)(p_1^2-p_2^2)(1-\exp [\beta(\gamma_2-\gamma_1)(p_1^2-p_2^2)])$ and $g(\p)>0$. Hence
by \eqref{21} and Lemma \ref{expineq}
\be 
\lbl{24}
(\gamma_2-\gamma_1)\left(\frac{a}{\gamma_3-\gamma_2}-\frac{b}{\gamma_3-\gamma_1}\right)<0.
\ee
But 
\be 
\lbl{25}\nonumber
(\gamma_2-\gamma_1)(a-b)&=&-\half\int_{S^2}p_3\exp(\alpha m_3p_3)f(\p)g(\p)\,d\p\\
&=&-\int_{\{p_3>0\}}p_3\sinh(\alpha m_3p_3)f(\p)g(\p)\,d\p>0,\nonumber
\ee 
since $g$ is even in $p_3$. A similar argument shows that $b>0$. Hence from \eqref{18}, \eqref{24} we get that
\be 
\lbl{26}\nonumber
(\gamma_2-\gamma_1)\left(\frac{1}{\gamma_3-\gamma_2}-\frac{1}{\gamma_3-\gamma_1}\right)b=\frac{(\gamma_2-\gamma_1)^2}{(\gamma_3-\gamma_1)(\gamma_3-\gamma_2)}b<0,
\ee 
a contradiction by \eqref{18a}.
\end{proof}

\subsection{Non-axisymmetric critical points}
\lbl{nonaxi}
In  the appendix to \cite{zhouetal2007} it is shown that in general critical points for the coupled dipolar/Maier-Saupe interaction are not axisymmetric. We now give a general sufficient condition for there to exist non-axisymmetric critical points (which, however, does not apply to the coupled dipolar/Maier-Saupe interaction -- see Remark \ref{notwork}).

We denote by $P^{48}$  the cubic group consisting of real  orthogonal  $3\times 3$ matrices $\Q_i, i=1,\ldots 48,$ with each row and each column having a single nonzero entry  $\pm 1$. A probability density $\rho\in \mathcal P$ has {\it cubic symmetry} provided
$$\rho(\Q_i \p)=\rho(\p) \text{  for all }i\text{ and a.e. }\p\in S^2.$$\begin{thm}
\lbl{nonaxithm}
There exists a non-axisymmetric critical point with cubic symmetry if
\be 
\lbl{27}
\int_{-1}^1k(t)P_{2r}(t)\,dt<-2\tau,
\ee
for some $r\geq 2$, where $P_{2r}$ is the $(2r)^{\rm th}$ Legendre polynomial.
\end{thm}

\begin{prop}
\lbl{axicubic}
If $\rho\in \mathcal P\cap C^0(S^2)$ is axisymmetric with respect to two nonparallel axes $\e,\bar\e\in S^2$ then $\rho(\p)=\frac{1}{4\pi}$ is constant. 
\end{prop}
\begin{proof}
By assumption we have that  $\rho(\p)=f(\p\cdot\e)=g(\p\cdot\bar\e)$ for some continuous functions $f,g:[-1,1]\to\R$.  Setting $\lambda:=\e\cdot\bar\e$ we have that $|\lambda|<1$ and $|\e\wedge\bar\e|^2=1-\lambda^2$. Given $s,t\in[-1,1]$ we choose
$$\p=\alpha\e+\beta\bar \e+\gamma\e\wedge\bar\e,$$
where
$$\alpha=\frac{s-t\lambda}{1-\lambda^2},\;\;\beta=\frac{t-s\lambda}{1-\lambda^2}.$$
Then $\p\cdot\e=s$, $\p\cdot\bar\e=t$, and $\p\in S^2$ provided
$$\gamma^2(1-\lambda^2)^2=-(s^2+t^2-2\lambda st+\lambda^2-1)>0.$$
Hence given $t\in [-1,1]$ we have that $f(s)=g(t)$ for all $s$ satisfying
$$s^2+t^2-2\lambda st+\lambda^2-1<0,$$
so that $f(s)$ is constant in the interval 
$$(\lambda  t-\sqrt{(1-\lambda^2)(1-t^2)}, \lambda t+\sqrt{(1-\lambda^2)(1-t^2)}).$$
This interval depends continuously on $t$ and at $t=\lambda$ the right-hand endpoint is $1$, while at $t=-\lambda$ the left-hand endpoint is $-1$. Therefore $f$ is constant on $[-1,1]$, giving the result.
\end{proof}
\begin{rem}
\lbl{soph}
A more sophisticated proof is to use the fact that any rotation is a product of a finite number of rotations about $\e$ and $\bar\e$, as proved in \cite{lowenthal1971} (see also \cite{hamada}). Thus $\rho(\Rv\p)=\rho(\p)$ for all $\Rv\in SO(3)$, so that $\rho$ is constant.
\end{rem}
\begin{cor}
\lbl{cubicaxi}
If $\rho\in \mathcal P\cap C^0(S^2)$ is axisymmetric and has cubic symmetry then  $\rho(\p)=\frac{1}{4\pi}$ is constant. 
\end{cor}
\begin{proof}
If $\rho$ is axisymmetric with axis $\e$ and has cubic symmetry then we have that 
$$\rho(\Q_i\p)=f(\Q_i\p\cdot\e)=f(\p\cdot\Q_i^T\e)\text{  for all }\Q_i\in P^{48}.$$
But if $\Q_i^T\e$ were parallel to $\e$ for all $i$ this would imply that $\Q_i^2\in P^{24}:=P^{48}\cap SO(3)$ has the same axis of rotation $\e$ for each $i$, which is not the case. Hence $\rho$ is axisymmetric with respect to two nonparallel axes of rotation.
\end{proof}

\begin{proof}[Proof of Theorem \ref{nonaxithm}]
The strategy of the proof is to show:\vspace{.03in}

\noindent (i)   that $I$ attains a minimum $\rho_c$ among probability densities having cubic symmetry,\vspace{.03in}

\noindent (ii)   that $\rho_c$ is a critical point of $I$; this is an example of the {\it principle of symmetric criticality} \cite{palais}, but for technical reasons and simplicity we give a direct proof,\vspace{.03in}

\noindent (iii)  that there is a $(2r)^{\rm th}$ order spherical harmonic $u$ having cubic symmetry, and  via the Funk-Hecke Theorem (see \eqref{funk} below) that the second variation  $\delta^2I(u,u)$ at the isotropic state $\rho_{\rm iso}(\p)=\frac{1}{4\pi}$ is negative for some $r\geq 2$ if \eqref{27} holds, so that $\rho_c$ is not isotropic,\vspace{.03in}

\noindent (iv)  that hence by Corollary \ref{cubicaxi}   $\rho_c$ is not axisymmetric.\smallskip

\noindent {\it Step} (i).  
The set
$${\mathcal A}:=\{\rho\in\mathcal P: \rho \text{ has cubic symmetry}\}$$
is nonempty (since $\rho_{\rm iso}\in \mathcal A$) and weakly closed in $L^1(S^2)$. Let $\rho^{(j)}$ be a minimizing sequence for $I$ in $\mathcal A$. By the de la Vall\'ee Poussin criterion \cite[Chapter II]{DellacherieMeyer} there therefore exists a subsequence, not relabelled, such that $\rho^{(j)}\weak \rho_c$ in $L^1(S^2)$ for some $\rho_c\in\mathcal A$. Then the convexity of $\rho\ln\rho$ and the weak continuity of the interaction term imply that $\rho_c$ is a minimizer.\smallskip

\noindent{\it Step }(ii). We proceed as in the proof of Proposition \ref{critprop}. The same argument as there shows that 
\be 
\label{jb5}
\int_{S^2}[\tau(1+\ln \rho_c)+(k*\rho_c)](\rho'- \rho_c)\,d\p\geq 0,
\ee
for any $\rho'\in\mathcal A$. Choosing   $\rho'=\rho_{\rm iso}$ we deduce that
$\rho_c(\p)>0$ for a.e. $\p\in S^2$. 

For $\ep>0$ define 
$$S_\ep=\{\p\in S^2:\ep<\rho_c(\p)<\ep^{-1}\}.$$
Let $u\in L^\infty(S_\ep)$ with $\int_{S_\ep}u(\p)\,d\p=0$ and define
\be 
\nonumber
\rho'(\p):=\left\{\begin{array}{ll}\rho_c(\p)+s\displaystyle\sum_{i=1}^{48}u(\Q_i\p),& \p\in S_\ep,\\
\rho_c(\p),&\p\not\in S_\ep,
\end{array}\right.
\ee
for $|s|$ sufficiently small. Then, from \eqref{jb5} we have that 
\be 
\lbl{30}
\int_{S_\ep}[\tau(1+\ln \rho_c(\p) )+(k*\rho_c)(\p)]\sum_{i=1}^{48}u(\Q_i\p)\,d\p= 0.
\ee
Making the change of variables $\q\to\Q_i\p$ and using the cubic invariance of $\rho_c$ we deduce from \eqref{30} that 
$$\int_{S_\ep}[\tau(1+\ln \rho_c(\p) )+(k*\rho_c)(\p)]u(\p)\,d\p= 0,$$
and thus deduce as before that $\rho_c$ is a critical point.\smallskip

\noindent  {\it Step }(iii). The study of spherical harmonics with cubic symmetry is a classical topic (see, for example, \cite{vanderlagebethe}). An example of a $(2r)^{th}$ order (not normalized) spherical harmonic with cubic symmetry is 
\be 
\lbl{harmonic}
u(\p)=P_{2r}(p_1)+P_{2r}(p_2)+P_{2r}(p_3),\;\;\p\in S^2,
\ee
whose cubic symmetry is obvious since $P_{2r}(t)$ is even in $t$, and which is a linear combination of $(2r)^{\rm th}$ order spherical harmonics because in spherical polar coordinates $Y^0_{2r}(\theta,\varphi)=\sqrt\frac{4r+1}{4\pi} P_{2r}(\cos \theta)$. Note that since $P_2(t)=\half(3t^2-1)$ we have that $u=0$ for $r=1$. To show that $u$ is nonzero for $r>1$ we take $\p=\e_1=(1,0,0)$, so that $u(\e_1)=1+2P_{2r}(0)$. But  
$$
P_{2r}(0)=\frac{(-1)^r(2r)!}{2^{2r}(r!)^2},
$$
so that
\be 
\lbl{31}
\left|\frac{P_{2(r+1)}(0)}{P_{2r}(0)}\right|=\left|\frac{2r+1}{2r+2}\right|<1.
\ee
Hence $|P_{2r}(0)|<|P_2(0)|=\half$ and thus $u(\e_1)>0$.

Since the integral of any nonconstant spherical harmonic over $S^2$ is zero, we have that $\int_{S^2}u(\p)\,d\p=0$. The Funk-Hecke Theorem (see, for example, \cite[Theorem 2.22]{atkinsonhan}) implies that for any $l^{\rm th}$ order spherical harmonic $Y^m_l$
\be 
\lbl{funk}
\int_{S^2}k(\p\cdot\q)Y^m_l(\q)\,d\q=2\pi\lambda_lY^m_l(\p),\;\; \p\in S^2,
\ee
where 
$\lambda_l:=\int_{-1}^1k(t)P_l(t)\,dt$ and $P_l(t)$ is the $l^{\rm th}$ Legendre polynomial. Thus we have that for $r\geq 2$
\be 
\lbl{32}
\int_{S^2}k(\p\cdot\q)u(\q)\,d\q=2\pi\lambda_{2r}u(\p), \;\;\p\in S^2.
\ee
But
\be 
\lbl{33}
\delta^2I(\rho_{\rm iso})(u,u)&:=&\left.\frac{d^2}{dt^2}I(\rho_{\rm iso}+tu)\right|_{t=0}\\
&=&2\pi(2\tau+\lambda_{2r})\int_{S^2}u(\p)^2\,d\p,\nonumber
\ee
so that if $\lambda_{2r}<-2\tau$ for some $r\geq 2$ then $\rho_c\neq \rho_{\rm iso}$.\smallskip

\noindent {\it Step} (iv).   By Corollary \ref{cubicaxi}, if $\rho_c$ were axisymmetric we would have $\rho_c=\rho_{\rm iso}$, which is impossible by Step (iii).

\end{proof}
\begin{rem}
\lbl{notwork}
Note that Theorem \ref{nonaxithm} does not give any information for the coupled dipolar/Maier-Saupe interaction, because both $t$ and $t^2$ are orthogonal to $P_{2r}(t)$ in $L^2(-1,1)$ for $r\geq 2$.
\end{rem}
As an example to which Theorem \ref{nonaxithm} applies we consider the Onsager potential $k(t)=\tau\sqrt{1-t^2}$. Writing 
$$\theta_r:=\int_{-1}^1 \sqrt{1-t^2}P_{2r}(t)\,dt,$$
we have, using the recurrence relation
$$(2r+2)P_{2r+2}(t)=(4r+3)tP_{2r+1}(t)-(2r+1)P_{2r}(t),
$$ that
\be 
\lbl{36}
2(r+1)\theta_{r+1}(t)=(4r+3)\int_{-1}^1t\sqrt{1-t^2}P_{2r+1}(t)\,dt-(2r+1)\theta_r.
\ee
Since
$$\frac{t^2-1}{2r+1}\frac{d}{dt}P_{2r+1}(t)=tP_{2r+1}(t)-P_{2r}(t),$$
we have that 
\be 
\lbl{37}2(r+2)\int_{-1}^1t\sqrt{1-t^2}\,P_{2r+1}(t)\,dt=(2r+1)\theta_r,
\ee
so that combining \eqref{36}, \eqref{37} we obtain
$$\frac{\theta_{r+1}}{\theta_r}=\frac{(2r+1)(2r-1)}{4(r+1)(r+2)}.$$
Since $\theta_0=\frac{\pi}{2}$, we have that $\theta_1=-\frac{\pi}{16}$ and
$$\theta_2=-2^{-7}\pi<\theta_r \text{ for all }r>2.$$
Therefore it is best to choose $r=2$ in Theorem \ref{nonaxithm}, and we deduce that there is a non-axisymmetric critical point with cubic symmetry provided $$\tau<2^{-8}\pi\kappa.$$

Vollmer \cite{vollmer2016} studies the bifurcation from the various eigenvalues of the Euler-Lagrange equation linearized around $\rho_{\rm iso}$
$$\tau=-\half \theta_r\kappa= \frac{((2r)!)^2\pi\kappa}{2^{4r+2}(2r-1)(r+1)(r!)^4},$$ (given by her in an equivalent form in \cite[Theorem 1]{vollmer2016}), showing that the bifurcation for $r=1$ is to axisymmetric critical points. A natural conjecture is that for $r>1$ there is bifurcation to non-axisymmetric critical points, which could probably be proved by working in a space of functions invariant with respect to suitable subgroups of $O(3)$, but we do not pursue this here.

The above examples leave open the question of whether for general kernels, and in particular for the Onsager kernel, local or global minimizers of $I$ are axisymmetric. In similar problems (see, for example, \cite{j22}) non-axisymmetric minimizers can arise from secondary bifurcations from a primary branch of axisymmetric solutions.
 
\section*{Acknowledgements}
\noindent I would like to thank the Isaac Newton Institute for Mathematical Sciences for support and hospitality during the 2019 programme {\it The Mathematical Design of New Materials}, during which I was a Simons Fellow, and    when  the research in this paper was begun. This work was supported by EPSRC grant number EP/R014604/1. I am  grateful to Ibrahim Fatkullin, Valeriy Slastikov, Pingwen Zhang, Hailiang Liu, Hui Zhang, Duvan Henao, John Toland, Michaela Vollmer and Paul Garrett for their interest and discussion.
\bibliography{gen2,balljourn,ballconfproc,ballprep}
\bibliographystyle{abbrv}

\end{document}